\documentclass[final,12pt]{article}
\usepackage{amsfonts,color,morefloats,pslatex,a4wide}
\usepackage{amssymb,amsthm,amsmath,latexsym,pslatex,cite}

\newtheorem{theorem}{Theorem}
\newtheorem{lemma}[theorem]{Lemma}

\newtheorem{example}[theorem]{Example}

\usepackage[para]{threeparttable}

\newcommand{\gf}{{\mathrm{GF}}}

\newcommand{\wt}{{\mathtt{wt}}}

\newcommand{\GRM}{{\mathrm{GRM}}}

\newcommand{\Z}{\mathbb{{Z}}}

\newcommand{\m}{\mathbb{M}}
\newcommand{\C}{{\mathcal{C}}}

\newcommand{\bc}{{\mathbf{c}}}
\newcommand{\bu}{{\mathbf{u}}}
\newcommand{\bv}{{\mathbf{v}}}

\newcommand{\0}{\textbf{0}}

\makeatletter

\newcommand{\Rmnum}[1]{\expandafter\@slowromancap\romannumeral #1@}
\makeatother

\begin{document} 

\title{A Generalization of the Tang-Ding Binary Cyclic Codes
\thanks{
Z. Sun's research was supported by The National Natural Science Foundation of China under Grant Number 62002093. S. Zhu's research was supported by The National Natural Science Foundation of China under Grant Number Nos. 12171134 and U21A20428.}}

\author{Zhonghua Sun\thanks{School of Mathematics, Hefei University of Technology, Hefei, 230601, Anhui, China. Email:  sunzhonghuas@163.com}, 
\and Ling Li \thanks{School of Mathematics, Hefei University of Technology, Hefei, 230601, Anhui, China. Email: lingli2310@163.com  }, \and Shixin Zhu \thanks{School of Mathematics, Hefei University of Technology, Hefei, 230601, Anhui, China. Email: zhushixin@hfut.edu.cn }
}

\maketitle

\begin{abstract} 
Cyclic codes are an interesting family of linear codes since they have efficient decoding algorithms and contain optimal codes as subfamilies. Constructing infinite families of cyclic codes with good parameters is important in both theory and practice. Recently, Tang and Ding [IEEE Trans. Inf. Theory, vol. 68, no. 12, pp. 7842--7849, 2022] proposed an infinite family of binary cyclic codes with good parameters. Shi et al. [arXiv:2309.12003v1, 2023] developed the binary Tang-Ding codes to the $4$-ary case. Inspired by these two works, we study $2^s$-ary Tang-Ding codes, where $s\geq 2$. Good lower bounds on the minimum distance of the $2^s$-ary Tang-Ding codes are presented. As a by-product, an infinite family of $2^s$-ary duadic codes with a square-root like lower bound is presented.  
 
\vspace*{.3cm}
\end{abstract}

\section{Introduction, motivations and objectives}\label{sec:1}

Let $q$ be a prime power, $\gf(q)$ be the finite field with $q$ elements and let $\gf(q)^n$ denote the $n$-dimensional {\it row vector space} over $\gf(q)$. Every nonempty subspace $\C$ of $\gf(q)^n$ is called a $q$-ary {\it linear code} of length $n$. If $\C$ has dimension $k$ and minimum distance $d$, $\C$ is called a $q$-ary $[n, k, d]$ linear code. If a lower bound on the minimum distance $d$ of the code $\C$ is lower bounded by $c \sqrt{n}$ for a fixed constant $c$ and fix $q$, this lower bound is then called a {\it square-root-like lower bound}. If any $(c_0,c_1,\ldots, c_{n-1})\in \C$ implies that $(c_{n-1}, c_0,\ldots, c_{n-2})\in \C$, $\C$ is called a $q$-ary {\it cyclic code} of length $n$. Define 
\begin{align*}
\Phi:~\gf(q)^n &\rightarrow \gf(q)[x]/(x^n-1)\\
(c_0,c_1,\ldots, c_{n-1}) &\mapsto	c_0+c_1x+\cdots+c_{n-1}x^{n-1}
\end{align*}
and $\Phi(\C)=\{ \Phi(\bc):~\bc \in \C  \}$. It is known that $\gf(q)[x]/(x^n-1)$ is a {\it principal ideal} ring. The code $\C$ is cyclic if and only if $\Phi(\C)$ is an {\it ideal} of $\gf(q)[x]/(x^n-1)$. In the following, we will identify $\C$ with $\Phi(\C)$ for any $q$-ary cyclic code $\C$ of length $n$. Let $\C=(g(x))$ be a $q$-ary cyclic code of length $n$, where $g(x)$ is monic and has the smallest degree among non-zero codewords of $\C$. Then $g(x)$ is called the {\it generator polynomial} and $h(x)=(x^n-1)/g(x)$ is referred to as the {\it check polynomial} of $\C$. Cyclic codes over finite fields have a cyclic structure, resulting in good encoding and decoding algorithms. Cyclic codes also contain many infinite families of optimal codes as subfamilies. When $(m, q-1)=1$, there exists a $q$-ary $[(q^m-1)/(q-1),(q^m-1)/(q-1)-m, 3]$ cyclic Hamming code \cite[Theorem 5.1.4]{HP2003}. Several infinite families of distance-optimal cyclic codes were documented in \cite{DH13,LKZL,LLHDT,SLNS,XCX}. Recently, several infinite families of $q$-ary ($q\in \{2, 3, 4\}$) cyclic codes with a square-root-like lower bound were constructed in \cite{CDLS,LLD,LLQ, STKS2023,SOS,sun23,TD22}. However, very limited results on $q$-ary ($q>3$) cyclic codes with a square-root-like lower bound are known. Our first objective is to construct an infinite family of $2^s$-ary cyclic codes with a square-root-like lower bound.


Let $\C \subset \gf(q)^n$, the dual of $\C$ is defined by 
\begin{equation*}
\C^\perp =\{\bu \in \gf(q)^n:\bu \bv^t=0,~{\rm for ~all}~\bv \in \C \}.
\end{equation*}
Two motivations of this paper are documented as follows: 
\begin{enumerate}
\item If $\C\cap \C^\perp=\{ \0\}$, $\C$ is called a {\it linear complementary dual} (LCD) code. Binary LCD codes play an important role in armoring implementations against side-channel attacks and fault injection attacks \cite{CG16}. Carlet et al. \cite{CMTQP} showed that if a $q$-ary $[n, k, d]$ code ($q>3$) exists there is a $q$-ary $[n, k, d]$ LCD code. According to \cite[Proposition 3]{CG16}, $2^s$-ary LCD codes can be transformed into binary LCD codes by expansion. For more information on LCD codes, we refer the reader to \cite{CMTQP,CH,LDL,SOXS}. Motivated by the importance of cyclic codes and LCD codes in both theory and practice, our second objective is to construct an infinite family of $2^s$-ary cyclic LCD codes with a square-root-like lower bound. 
\item  If $\C=\C^\perp$, $\C$ is called a {\it self-dual} code.	 Self-dual codes have a rich mathematical theory, and are related to lattices \cite{leech}, invariant theory \cite{NRS} and $t$-designs \cite{DT}. Sloane and Thompson characterized binary cyclic self-dual codes \cite{ST}. Nedeloaia \cite{Nede}, Heijne and Top \cite{HT2009} constructed binary cyclic self-dual codes. Recently, Chen \cite{CHEN} constructed several infinite families of $2^s$-ary cyclic self-dual codes with a square-root like lower bound. Kai et al. \cite{KZ08} and Jia et al. \cite{JLX} independently proved that $q$-ary cyclic self-dual codes of length $n$ exist if and only if $n$ is even and $q=2^s$ with a positive integer $s$. For more information on the theory of self-dual codes, we refer the reader to \cite[Chapter 9]{HP2003} and \cite[Chapter 19]{MS77}. Motivated by the importance of self-dual codes in both theory and practice, our third objective is to construct an infinite family of self-dual codes with a square-root-like lower bound by extending a family of cyclic codes.       	
\end{enumerate}



Let $n$ be a positive integer and $\gcd(n, q)=1$. Let $\Z_{n}=\left\{0,1,2,\cdots,n-1 \right\}$ be the ring of integers modulo $n$. For any integer $i$, let $i \bmod n$ denote the unique integer $j$ such that $0 \leq j \leq n-1$ and $i-j$ is divisible by $n$ throughout this paper. For any $i \in \Z_{n}$, the \emph{$q$-cyclotomic coset of $i$ modulo $n$} is defined by 
\[C^{(q,n)}_i=\left\{i q^j \bmod {n}: 0\leq j\leq \ell_i-1 \right\}  \subset \Z_{n}, \]
where $\ell_i$ is the smallest positive integer such that $i \equiv i q^{\ell_i} \pmod{n}$. The smallest integer in $C_i^{(q, n)}$ is called the {\it coset leader} of $C_i^{(q, n)}$. Let $\Gamma$ be the set of all the coset leaders, then
\begin{equation*}
\Z_n=\bigcup_{i \in \Gamma} C_i^{(q, n)}	
\end{equation*}
and $C_i^{(q, n)} \cap C_j^{(q, n)}=\emptyset$ for any two distinct elements $i$ and $j$ in $\Gamma$. Let $\beta$ be a primitive $n$-th root of unity and let $\m_{\beta^i}(x)$ denote the minimal polynomial of $\beta^i$ over $\gf(q)$, then
\begin{equation*}
	x^n-1=\prod_{i\in \Gamma} \m_{\beta^i}(x),
\end{equation*}
which is the factorization of $x^n-1$ over $\gf(q)$. For a $q$-ary nonzero cyclic code $\C=(g(x))$ of length $n$, the set $T(\C)=\{i \in \Z_n:~g(\beta^i) =0 \}$ is called the {\it defining set} of $\C$ with respect to the primitive $n$-th root of unity $\beta$. Clearly, the defining set $T(\C)$ is the union of some $q$-cyclotomic cosets modulo $n$. 

Let $\emptyset \neq S\subset \Z_n$ and $v \in \Z_n $, define $v S=\{v i \bmod{n}:~i \in S \}$. Let $S_1$ and $S_2$ be two subsets of $\Z_n$ such that 
\begin{enumerate}
\item $S_1\cap S_2 =\emptyset$ and $S_1 \cup S_2 =\Z_n \backslash \{0 \}$.
\item Both $S_1$ and $S_2$ are the union of some $q$-cyclotomic cosets modulo $n$.	
\end{enumerate}
If there is a unit $v\in \Z_n$ such that $v S_1=S_2$ and $v S_2 =S_1$, then $(S_1, S_2, v)$ is called a {\it splitting} of $\Z_n$. Let $(S_1,S_2, v)$ be a splitting of $\Z_n$. The pair of $q$-ary cyclic codes $\C_1$ and $\C_2$ of length $n$ with defining set $S_1$ and $S_2$ with respect to the $n$-th primitive root of unity $\beta$ are called {\it odd-like duadic codes}, and the pair of $q$-ary cyclic codes $\widetilde{\C_1}$ and $\widetilde{\C_2}$ of length $n$ with defining set $S_1\cup \{0\}$ and $S_2\cup \{0\}$ with respect to the $n$-th primitive root of unity $\beta$ are called {\it even-like duadic codes}.


Duadic codes are a family of interesting cyclic codes that contain quadratic residue codes as subfamilies. The binary duadic codes were introduced by Leon, Masley and Pless \cite{LMP}, and their properties were studied in \cite{PML}. Arbitrary duadic codes were studied by Smid \cite{smid}. Currently, some infinite families of duadic codes have been constructed using cyclotomy and other methods (see \cite{D2012,DP99,LLQ,STKS2023,TNH,TD22}). Odd-like duadic codes have the following interesting properties. 
\begin{enumerate}
\item If $-1$ gives the splitting for a pair of odd-like duadic codes $\C_1$ and $\C_2$,	the extended codes of $\C_1$ and $\C_2$ are self-dual \cite[Theorem 6.4.12]{HP2003}.
\item The minimum odd-like weight $d_{o}$ of an odd-like duadic code of length $n$ obeys $d_{o}\geq \sqrt{n}$ \cite[Theorem 6.5.2]{HP2003}.  
\end{enumerate}
It is well-known that odd-like quadratic residue codes have a square-root-like lower bound. Tang and Ding \cite{TD22} constructed an infinite family of binary odd-like duadic codes of length $2^m-1$ with a square-root-like lower bound, where $m\geq 3$ is odd. Inspired by their work, Liu et al. \cite{LLQ} constructed another infinite family of binary odd-like duadic codes of length $2^m-1$ with a square-root-like lower bound, and Shi et al. \cite{STKS2023} constructed an infinite family of quaternary odd-like duadic codes of length $4^m-1$ with a square-root-like lower bound. According to \cite[Chapter 6]{HP2003}, the minimum distance of odd-like duadic codes may be even-like, and the minimum distance of many duadic codes is poor. Our fourth objective is to construct an infinite family of $2^s$-ary duadic codes of length $2^{sm}-1$ with a square-root-like lower bound. Our results to some extent answer a open problem raised by Shi et al \cite{STKS2023}.  

The rest of this paper is organized as follows. In Section \ref{sec:2}, we determine the dimension of the $2^s$-ary Tang-Ding codes. In Section \ref{sec:3}, we lower bound the minimum distance of the $2^s$-ary Tang-Ding codes. In Section \ref{sec:4}, we conclude this paper and make some concluding remarks.

\section{The $2^s$-ary Tang-Ding codes}\label{sec:2}

Throughout this paper, let $q=2^s$, where $s$ is a positive integer. Let $n=q^m-1$, where $m\geq 2$ is an integer. Let $\beta$ be a primitive element of $\gf(q^m)$. Below we introduce the $q$-ary Tang-Ding codes, which is a generalization of \cite{STKS2023} and \cite{TD22}.

Let $0\leq i\leq n$, the {\it $q$-adic expansion} of $i$ is defined to be $i=\sum_{j=0}^{m-1} i_j q^j$, where $0\leq i_j \leq q-1$. The {\it $q$-weight} of $i$, denoted by $\wt_q(i)$, is defined to be $\sum_{j=0}^{m-1} i_j$. It is easily verified that $\wt_q(i)$ is a constant on each cyclotomic coset $C_j^{(q, n)}$. Define 
\begin{equation*}
T_{(q, m;0)}=\{ 1\leq i\leq n-1:~\wt_q(i) \equiv 0\pmod{2}\}	
\end{equation*}
and 
\begin{equation*}
T_{(q, m;1)}=\{ 1\leq i\leq n-1:~\wt_q(i) \equiv 1\pmod{2}\}	.
\end{equation*}
It is easily verified that the following hold.
\begin{enumerate}
\item $T_{(q, m; 0)}$ and $T_{(q, m; 1)}$ are the union of some $q$-cyclotomic coset modulo $n$.
\item $T_{(q, m; 0)}\cap T_{(q, m; 1)}=\emptyset$.
\item $0\notin T_{(q, m; 0)}\cup T_{(q, m; 1)}$.
\item $\Z_n=\{0\} \cup T_{(q, m; 0)} \cup T_{(q, m; 1)} $.
\end{enumerate}
For each $i \in \{0,1\}$, let $\C_{(q, m; i)}$ (resp. $\widetilde{\C_{(q, m; i)}}$)  be the $q$-ary cyclic code with length $n$ and defining set $T_{(q, m; i)}$ (resp. $T_{(q, m; i)} \cup \{0\}$) with respect to the $n$-th primitive root of unity $\beta$. The pair of codes $\C_{(q, m; 0)}$ and $\C_{(q, m; 1)}$ are called $q$-ary Tang-Ding codes. Specifically, if $q=2$, the codes $\C_{(q, m; 0)}$ and $\C_{(q, m; 1)}$ were first studied by Tang and Ding \cite{TD22}; if $q=4$, the codes $\C_{(q, m; 0)}$ and $\C_{(q, m; 1)}$ were studied by Shi et al. \cite{STKS2023}. To study the $q$-ary Tang-Ding codes, we need the following lemma.

\begin{lemma}\label{lem:1}
Let $q=2^s$, $n=q^m-1$ and $m\geq 2$.	The following hold.
\begin{enumerate}
\item If $m\geq 3$ is odd, $- T_{(q, m; 0)}=T_{(q, m;1)}$ and $|T_{(q, m; 0)}|=|T_{(q, m;1)}|=(n-1)/2$.
\item If $m\geq 2$ is even, 	$- T_{(q, m; 0)}=T_{(q, m;0)}$ and $- T_{(q, m; 1)}=T_{(q, m;1)}$. Furthermore, $|T_{(q, m; 0)}|=(n-3)/2$ and $|T_{(q, m; 1)}|=(n+1)/2$.
\end{enumerate}
\end{lemma}

\begin{proof}
It is clear that $\wt_q( n-i)=(q-1) m -\wt_q(i)$ for each $0\leq i\leq n-1$.	 Since $q$ is even, 
\begin{align*}
\wt_q( n-i)&\equiv  m -\wt_q(i) \pmod{2}	\\
&\equiv \begin{cases}
1-\wt_q(i) \pmod{2} &~{\rm if~}m~{\rm is ~odd},\\
\wt_q(i)\pmod{2}&~{\rm if~}m~{\rm is ~even}. 	
\end{cases}
\end{align*}
Therefore, the following hold. 
\begin{enumerate}
\item If $m$ is odd, $i\in T_{(q, m; 0)}$ if and only if $n-i\in T_{(q, m; 1)}$. Consequently, $-T_{(q, m; 0)}=T_{(q, m; 1)}$ and $$|T_{(q, m; 0)}|=|T_{(q, m; 1)}|=(n-1)/2.$$ 
\item If $m$ is even, $j\in T_{(q, m; i)}$ if and only if $n-j\in T_{(q, m; i)}$ for each $i \in \{0,1\}$. The first desired conclusion of Result 2 follows.	
\end{enumerate}

We now determine $|T_{(q, m; 0)}|$ and $|T_{(q, m; 1)}|$. Clearly, $|T_{(q, m; 0)}|+|T_{(q, m; 1)}|=n-1$. It is easily verified that 
\begin{align*}
|T_{(q, m; 1)}|&=|\{(i_0,i_1,\ldots, i_{m-1})\in \Z_{q}^m:~i_0+i_1+\cdots+i_{m-1}\equiv 1\pmod{2}\}|\\
&=|\{(i_0,i_1,\ldots, i_{m-2})\in \Z_{q}^{m-1}:~i_0+i_1+\cdots+i_{m-2}~{\rm is ~odd~and} ~0\leq i_{m-1}\leq q-1~{\rm is~even}\}|\\
&~~~+|\{(i_0,i_1,\ldots, i_{m-2})\in \Z_{q}^{m-1}:~i_0+i_1+\cdots+i_{m-2}~{\rm is ~even~and} ~0\leq i_{m-1}\leq q-1~{\rm is~odd}\}|\\
&=\frac{q}2\times |\{(i_0,i_1,\ldots, i_{m-2})\in \Z_{q}^{m-1}:~i_0+i_1+\cdots+i_{m-2}~{\rm is ~odd}\}|+\\
&~~~~\frac{q}2\times |\{(i_0,i_1,\ldots, i_{m-2})\in \Z_{q}^{m-1}:~i_0+i_1+\cdots+i_{m-2}~{\rm is ~even}\}|\\
&=\frac{q^m}2.
\end{align*}
Consequently, $|T_{(q, m, 0)}|=(n-3)/2$. This completes the proof.  
\end{proof}

For any $\C \subset \gf(q)^n$, the {\it extended code} of $\C$ is defined by
\begin{equation*}
	\overline{\C}=\left\{(\bc, c_\infty):~\bc\in \C~{\rm and~}c_{\infty}=- \sum_{i=0}^{n-1}c_i \right \}.
\end{equation*}
 A {\it complement} of $\C$ is a code $\C^c$ such that $\C+\C^c=\gf(q)^n$ and $\C\cap \C^c=\{\0\}$. By definition, $\C_{(q, m; 0)}^c=\widetilde{\C_{(q, m; 1)}}$ and $\C_{(q, m; 1)}^c=\widetilde{\C_{(q, m; 0)}}$. 
The more properties of the codes $\C_{(q, m; i)}$, $\widetilde{\C_{(q, m; i)}}$ and $\overline{\C_{(q, m; i)}}$  are documented in the following theorems.

\begin{theorem}\label{thm:2}
Let $m\geq 3$ be odd, $q=2^s$ and $n=q^m-1$. The following hold.
\begin{enumerate}
\item The codes $\C_{(q, m; 0)}$ and $\C_{(q, m; 1)}$ form a pair of odd-like duadic codes with length $n$ and dimension $(n+1)/2$.	
\item The extended code $\overline{\C_{(q, m; i)}}$ is a $q$-ary self-dual code of length $n+1$.
\item The code $\widetilde{\C_{(q, m; i)}}$ is a $q$-ary self-orthogonal code with length $n$ and dimension $(n-1)/2$.
\item The codes $\C_{(q, m, 0)}^\perp$ and $\widetilde{\C_{(q, m; 1)}}$ have the same parameters, and the codes $\C_{(q, m, 1)}^\perp$ and $\widetilde{\C_{(q, m; 0)}}$ have the same parameters, respectively. 
\end{enumerate}
\end{theorem}

\begin{proof}
By Result 1 of Lemma \ref{lem:1}, $-1$ gives the splitting of $\C_{(q, m; 0)}$ and $\C_{(q, m; 1)}$. The desired result 1 follows. Notice that $1+n=0$ in $\gf(q)$, the desired result 2 directly follows from \cite[Theorem 6.4.12]{HP2003}. The desired result 3 directly follows from \cite[Theorem 6.4.1]{HP2003}. For any cyclic code $\C$, by\cite[Theorem 4.49]{HP2003}, we get that $\C^\perp$ and $\C^c$ have the same parameters. Notice that $\C_{(q, m; 0)}^c=\widetilde{\C_{(q, m; 1)}}$ and $\C_{(q, m; 1)}^c=\widetilde{\C_{(q, m; 0)}}$, the desired result 4 follows. This completes the proof. 
\end{proof}

\begin{theorem}\label{thm:3}
Let $m\geq 2$ be even, $q=2^s$ and $n=q^m-1$. The following hold.
\begin{enumerate}
\item $\C_{(q, m; 0)}^\perp=\widetilde{\C_{(q, m; 1)}}$ and $\C_{(q, m; 1)}^\perp=\widetilde{\C_{(q, m; 0)}}$. 

\item The code $\C_{(q, m; 0)}$ is a $q$-ary cyclic LCD code with length $n$ and dimension $(n+3)/2$.
\item The code $\C_{(q, m; 1)}$ is a $q$-ary cyclic LCD code with length $n$ and dimension $(n-1)/2$.  
\end{enumerate}
\end{theorem}

\begin{proof}
According to \cite[Theorem 4.4.9]{HP2003}, $\C_{(q, m; i)}^\perp$ has the defining set $T_{(q, m; i)}^\perp:=\Z_n \backslash -T_{(q, m; i)}$ with respect to the $n$-th primitive root of unity $\beta$ for each $i\in \{0,1\}$. By Result 2 of Lemma \ref{lem:1}, $-T_{(q, m ;i)}=T_{(q, m; i)}$. Consequently, $T_{(q, m; 0)}^\perp=\{0\} \cup T_{(q, m; 1)}$ and $T_{(q, m; 1)}^\perp=\{0\} \cup T_{(q, m; 0)}$. The desired result 1 follows from the definitions of the codes $\C_{(q, m; i)}$ and $\widetilde{\C_{(q, m; i)}}$. Notice that $T_{(q, m; i)}\cap T_{(q, m; i)}^\perp=\emptyset$, we obtain that $\C_{(q, m; i)}$ is a LCD code. The dimension of the code $\C_{(q, m; i)}$ follows from Result 2 of Lemma \ref{lem:1}. This completes the proof. 
\end{proof}

\section{The minimum distance of the $2^s$-ary Tang-Ding codes}\label{sec:3}

In this section, we lower bound the minimum distances of the $2^s$-ary Tang-Ding codes and their related codes. The minimum distance of the binary Tang-Ding codes was studied in \cite{TD22}. From now on, we always suppose that $q=2^s\geq 4$. We need the following lemmas.

\begin{lemma}[The BCH bound]\cite[Corollary 9]{MS77}\label{lem:4}
Let $\C$ be a $q$-ary cyclic code with length $n$ and defining set $T(\C)$. If there are integer $h$, integer $b$, integer $a$ with $\gcd(a, n)=1$ and integer $\delta$ with $2\leq \delta \leq n$ such that $(b+ai)\bmod n \in T(\C)$ for each $h\leq i\leq h+\delta-2$, then $d(\C)\geq \delta$, where $d(\C)$ denotes the minimum distance of $\C$.
\end{lemma}

\begin{lemma}\cite[Lemma 3]{LDL} \label{lem:5}
	Let $q=2^s$ with a positive integer $s$. Let $m\geq 2$ and let $\ell$ be a positive integer such that $m/\gcd(\ell, m)$ is odd. Then $\gcd(q^m-1, q^\ell+1)=1$. 
\end{lemma}

\begin{lemma}\label{lem:6}
	Let $m\geq 2$, $q\geq 3$ and $n=q^m-1$. Let $2\leq A\leq q-1$ and $0\leq h\leq m-1$. Then
	\begin{equation*}
	\wt_q(A q^h-1-i)=(q-1)h+A-1-\wt_q(i)	
	\end{equation*}
	for each $0\leq i\leq A q^{h}-1$.
\end{lemma}

\begin{proof}
Suppose $i=B q^h+j$, where $0\leq B \leq A-1$ and $0\leq j\leq q^h-1$. Then 
$$A q^h-1-i= (A-1-B)q^h+q^h-1-j. $$
It follows that
\begin{align*}	
\wt_{q}(A q^{h}-1-i)&= A-1-B+\wt_q(q^h-1-j) \\
		&=(q-1)h+A-1-B-\wt_q(j)\\
		&= (q-1)h+A-1-\wt_q(i).
\end{align*}	
This completes the proof.
\end{proof}

When $m\geq 3$ is odd, to lower bound the minimum distance of $\C_{(q, m; i)}$, the following lemma is useful. 

\begin{lemma}\label{lem:7}
Let $m\geq 3$ be odd, $q=2^s\geq 4$ and $n=q^m-1$. Let $b=2 q^{m-1}$ and $a=q^{(m-1)/2}+1$. Then $\gcd(a, n)=1 $  and 
	\begin{equation*}
	\{b+ai:~-(q-1)\leq i\leq q^{(m-1)/2}+q-2\}\subseteq T_{(q, m; 0)}. 
	\end{equation*}
\end{lemma}

\begin{proof}
Clearly, $\gcd((m-1)/2, m)=1$. By Lemma \ref{lem:5}, $\gcd(a, n)=1$. The rest of the proof of this lemma is divided into the following three cases.

{\it Case 1}: $0\leq i \leq q^{(m-1)/2}-1$. In this case, $b+ai =2q^{m-1}+iq^{(m-1)/2}+i$. It follows that 
$$
\wt_q( b+ai)=2+2\wt_q(i)	 \equiv 0\pmod{2}.
$$

{\it Case 2}: $i=q^{(m-1)/2}+j$, where $0\leq j\leq q-2$. In this case, 
	$$
	b+ai =3q^{m-1}+(j+1)q^{(m-1)/2}+j.
$$
It follows that $ \wt_{q}( b+ ai)=4+2j \equiv 0\pmod{2}$.

{\it Case 3}: $i=-j$, where $1\leq j\leq q-1$. In this case,
$$b+ai=q^{(m-1)/2}(2q^{(m-1)/2}-1-j)+q^{(m-1)/2}-j.$$
It follows that 	
	\begin{align*}
	\wt_{q}(b+a i)&=\wt_{q}(2q^{(m-1)/2}-1-j)+\wt_{q}(q^{(m-1)/2}-j)\\
	&=(q-1)(\tfrac{m-1}2)+1-\wt_q(j)+(q-1)(\tfrac{m-1}2)-\wt_q(j-1)\\
	&=(q-1)(m-1)+1-\wt_q(j)-\wt_q(j-1)\\
	&=(q-1)(m-1)+2-2j\\
	&\equiv 0\pmod{2},
\end{align*}	
where the second equation follows from Lemma \ref{lem:5}. 

Collecting the conclusions in Cases 1, 2 and 3 yields $\wt_{q}(b+ai)\equiv 0\pmod{2}$ for each $-(q-1)\leq i\leq q^{(m-1)/2}+q-2$. This completes proof.
\end{proof}

By the BCH bound and Lemma \ref{lem:7}, the minimum distance of the pair of odd-like duadic codes $\C_{(q, m; 0)}$ and $\C_{(q, m; 1)}$ has the following lower bound.  

\begin{theorem}\label{thm:8}
Let $m\geq 3$ be odd, $q=2^s\geq 4$ and $n=q^m-1$. Then 
$$d(\C_{(q, m; 0)})=d(\C_{(q, m; 1)})\geq q^{(m-1)/2}+2q-1.$$	
\end{theorem}

\begin{proof}
By Theorem \ref{thm:2}, the codes $\C_{(q, m; 0)}$ and $\C_{(q, m; 1)}$ form a pair of odd-like duadic codes. Hence, $d(\C_{(q, m; 0)})=d(\C_{(q, m; 1)})$. The lower bound on the minimum distance of $\C_{(q, m; 0)}$ directly follows from Lemma \ref{lem:7} and the BCH bound. This completes the proof.  	
\end{proof}


To lower bound minimum distances of the codes $\C_{(q, m ;0)}$ and $\C_{(q, m ;1)}$ for $m\equiv 2\pmod{4}$ and $m\geq 2$, we need the following lemmas.

\begin{lemma}\label{lem:9}
Let $m\equiv 2\pmod{4}$ and $m\geq 6$. Let $q=2^s\geq 4$, $n=q^m-1$, $b=q^{m-2}$ and $a=q^{(m-2)/2}+1$. Then $\gcd(a,n)=1 $  and 
$$
\{b+ai:-(q-1)\leq i\leq q^{(m-2)/2}+q-2\}\subseteq T_{(q,m;1)}. 
$$
\end{lemma}

\begin{proof}
Since $m\equiv 2\pmod{4}$, $m/\gcd((m-2)/2,m)$ is odd. By Lemma \ref{lem:5}, $\gcd(a, n)=1$. Similar to Lemma \ref{lem:7}, we can prove that $\wt_{q}(b+ai)\equiv 1\pmod{2}$ for each $-(q-1)\leq i\leq q^{(m-2)/2}+(q-2)$. This completes proof.
\end{proof}


\begin{lemma}\label{lem:10}
Let $m\equiv 2\pmod{4}$ and $m\geq 10$. Let $q=2^s\geq 4$, $n=q^m-1$, $b=q^{m-1}+q^{m-2}$ and $a=q^{(m-2)/2}+1$. Then $\gcd(a,n)=1 $  and 
$$
\{b+ai:-(q-1)\leq i\leq q^{(m-2)/2}+q-2\}\subseteq T_{(q, m; 0)}. 
$$
\end{lemma}

\begin{proof}
It is easily verified that 
$$
\wt_q(q^{m-1}+q^{m-2}+(q^{(m-2)/2}+1)i  )=1+\wt_q(q^{m-2}+(q^{(m-2)/2}+1)i  )
$$
for any $-(q-1)\leq i\leq q^{(m-2)/2}+q-2$. The desired result directly follows from Lemma \ref{lem:9}. This completes the proof. 
\end{proof}

\begin{lemma}\label{lem:11}
Let $m=6$, $q=2^s\geq 4$, $n=q^m-1$, $b=q^{m-1}+q$ and $a=\frac{q^{m}-1}{q-1}-q^{(m+2)/2}-1$. Then $\gcd(a,n)=1 $  and 
$$
\{(b+ai) \bmod n:~-(q^{2}-q)\leq i\leq q^{2}-q \}\subseteq T_{(q,m; 0)}. 
$$
\end{lemma}

\begin{proof}
It is clear that $a\equiv m-2\equiv 4 \pmod{q-1}$, then $\gcd(a, q-1 )=1$. Consequently,
\begin{align*}
\gcd(a, n)&=\gcd(a, n/(q-1))\\
&=\gcd(q^{(m+2)/2}+1, n/(q-1))\\
&=\gcd(q^{(m+2)/2}+1, q^m-1)\\
&=1.	
\end{align*}
Below we will prove that $\wt_q((b+ai)\bmod{n})\equiv 0\pmod{2}$ for $-(q^2-q)\leq i\leq q^2-q$ in the following two cases 

{\it Case 1}: $0\leq i\leq q^2-q$. Suppose $i=(q-1)i_1+i_0$, where $0\leq i_0, i_1\leq q-1$, then
\begin{equation}\label{EQ::1}
b+ai\equiv (i_0-i_1+1)q^5+i_1q^4+i_0q^3+i_0q^2+(i_0-i_1+1)q+i_1 \pmod{n}.	
\end{equation}
The rest of the proof is divided into the following subcases.
\begin{itemize}
\item If $i_0-i_1+1= q$, i.e., $i_0=q-1$ and $i_1=0$, then $(b+ai)\bmod n=q^4+1$. Consequently, 
$$\wt_q((b+a i ) \bmod n )=2.$$
\item If $0\leq i_0-i_1+1\leq q-1$, by Equation (\ref{EQ::1}),
$$(b+ai) \bmod n= (i_0-i_1+1)q^5+i_1q^4+i_0q^3+i_0q^2+(i_0-i_1+1)q+i_1.$$ 
It follows that $\wt_q((b+a i ) \bmod n )=4i_0+2$. 	
\item If $i_0=0$ and $i_1>1$, by Equation (\ref{EQ::1}), we get that 
$$
(b+ai)q \bmod n= (i_1-1)q^5+(q-1)q^4+(q-1)q^3+(q+1-i_1)q^2+(i_1-1)q+q+1-i_1.
$$
It follows that 
\begin{align*}
\wt_q((b+ai) \bmod n)&=\wt_q((b+ai)q \bmod n)\\
&=4q-2.
\end{align*}
\item If $i_0>0$ and $i_1>i_0+1$, by Equation (\ref{EQ::1}), we get that 
\begin{align*}
(b+ai)q \bmod n= i_1q^5+i_0q^4+(i_0-1)q^3+(q+i_0-i_1+1)q^2+(i_1-1)q+q+i_0-i_1+1.	
\end{align*}
It follows that 
\begin{align*}
\wt_q((b+ai) \bmod n)&=\wt_q((b+ai)q \bmod n)\\
&=2q+4i_0.
\end{align*}
\end{itemize}

{\it Case 2}: $-1\leq i\leq -(q^2-q)$. Suppose $-i=(q-1)i_1+i_0$, where $0\leq i_0,i_1\leq q-1$, then
\begin{equation}\label{EQ::2}
-(b+ai)\equiv (i_0-i_1-1)q^5+i_1q^4+i_0q^3+i_0q^2+(i_0-i_1-1)q+i_1 \pmod{n}.		
\end{equation}
The rest of the proof is divided into the following subcases.
\begin{itemize}
\item If $i_0-i_1-1\geq 0$, by Equation (\ref{EQ::2}), we get that  
\begin{align*}
\wt_q( (b+ai)\bmod n )&=(q-1)m-\wt_q( -(b+ai)\bmod n )\\
&=6(q-1)-4i_0+2.
\end{align*}

\item If $i_0-i_1-1<0$, 	it follows from $i\neq 0$ that $1\leq i_1\leq q-1$ and $0\leq i_0\leq i_1\leq q-1$. By Equation (\ref{EQ::2}), we get that 
$$-(b+ai)q \equiv i_1q^5+i_0q^4+i_0q^3+(i_0-i_1-1)q^2+i_1q+(i_0-i_1-1) \pmod{n}.$$
It follows that 
\begin{align*}
\wt_q( -(b+ai) \bmod{n})&=\wt_q( -(b+ai)q \bmod{n})\\
&=\begin{cases}
4q-6~&{\rm if}~i_0=0,\\
2q-4+4i_0~&{\rm if}~i_0>0.	
\end{cases}
\end{align*}
Consequently,
 \begin{align*}
\wt_q( (b+ai)\bmod n )&=(q-1)m-\wt_q( -(b+ai)\bmod n )\\
&=\begin{cases}
2q~&{\rm if}~i_0=0,\\
4q-2-4i_0~&{\rm if}~i_0>0.	
\end{cases}
\end{align*}
\end{itemize}

Collecting the conclusions in Cases 1 and 2 yields $\wt_{q}((b+ai)\bmod n )\equiv 0\pmod{2}$ for each $-(q^{2}-q)\leq i\leq q^{2}-q$. This completes proof.
\end{proof}

By the BCH bound and Lemmas \ref{lem:9}, \ref{lem:10} and \ref{lem:11}, we obtain the following results.

\begin{theorem}\label{thm:12}
Let $m\equiv 2\pmod{4}$ and $m\geq 2$. Let $q=2^s\geq 4$ and $n=q^m-1$. Then 
\begin{align*}
d(\C_{(q, m; 0)})\geq \begin{cases}
 (q+2)/2~&{\rm if}~m=2,\\
 2q^2-2q+2~&{\rm if}~m=6,\\
 q^{(m-2)/2}+2q-1~&{\rm if}~m\geq 10.
 \end{cases}	
\end{align*}
and 
\begin{align*}
d(\C_{(q, m; 1)})\geq \begin{cases}
 (q+2)/2~&{\rm if}~m=2,\\	
 q^{(m-2)/2}+2q-1~&{\rm if}~m\geq 6.	
 \end{cases}	
\end{align*} 	
\end{theorem}

\begin{proof}
When $m=2$ and $0\leq i\leq q/2-1$, we have $\wt_q(2q+2i)=2i+2\equiv 0\pmod{2}$ and  
 $$\wt_q(q+2i)=2i+1\equiv 1\pmod{2}.$$ 
By the BCH bound, $d(\C_{(q, m; 0)})\geq (q+2)/2$ and $d(\C_{(q, m; 1)})\geq (q+2)/2$. When $m\geq 6$, the desired lower bounds on the minimum distances of $\C_{(q, m; 0)}$ and $\C_{(q, m; 1)}$ follow from the BCH bound and Lemmas \ref{lem:9}, \ref{lem:10} and \ref{lem:11}. This completes the proof. 
\end{proof}

To lower bound minimum distances of the codes $\C_{(q, m ;0)}$ and $\C_{(q, m ;1)}$ for $m\equiv 0\pmod{4}$ and $m\geq 4$, we need the following lemmas.

\begin{lemma}\label{lem:13}
	Let $m \equiv 0\pmod{4}$ and $m\geq 4$. Let $q=2^s\geq 4$, $n=q^m-1$, $b=q^{m-1}$ and $a=(\frac{q-2}{2})(\frac{q^{m}-1}{q-1})+\frac{q^{(m-2)/2}-1}{q-1}$. Then $\gcd(a,n)=1 $  and 
	$$
	\{ (b+ai) \bmod n:~1\leq i\leq q^{(m-2)/2} \} \subseteq T_{(q,m; 0)}. 
	$$
\end{lemma}
\begin{proof}
Clearly, $a \equiv(\frac{q}{2}-1)m+\frac{m}{2}-1 \equiv -1\pmod{q-1}$. Hence, $\gcd(a, q-1)=1$. Consequently, 
	\begin{align*}
	\gcd(a, q^{m}-1) &= \gcd(a,\tfrac{q^{m}-1}{q-1})\\
	&= \gcd(\tfrac{q^{{(m-2)/2}}-1}{q-1},\tfrac{q^{m}-1}{q-1})\\
	&= \frac{\gcd(q^{(m-2)/2}-1,q^m-1)}{q-1} \\
	&=\frac{q^{\gcd((m-2)/2, m)}-1}{q-1}\\
	&=1.
	\end{align*}

Let $ i=(q-1)i_{1}+i_{0}$, where $ 0\leq i_{0}\leq q-2$, $ 0\leq i_{1}\leq \frac{q^{(m-2)/2}-1}{q-1}$ and $i_0\in \{0,1\}$ if $i_1=\frac{q^{(m-2)/2}-1}{q-1}$. It is easily verified that
	\begin{equation}\label{EQ::3}
	b+ai \equiv q^{m-1}+(q^{(m-2)/2}-1)i_1+(\tfrac{q-2}{2})(\tfrac{q^m-1}{q-1})i_0+(\tfrac{q^{(m-2)/2}-1}{q-1})i_{0} \pmod{n}.
	\end{equation} 
The rest of the proof is divided into the following three cases.

{\it Case 1}: $i_0=0$ and $ 1\leq i_{1}\leq \frac{q^{(m-2)/2}-1}{q-1}$. By Equation (\ref{EQ::3}), we obtain that 
$$(b+ai) \bmod{n}=q^{m-1}+(q^{(m-2)/2}-1)i_1.$$
It follows that 
$$\wt_q((b+ai) \bmod{n})=1+(q-1)(\tfrac{m-2}2)\equiv 0\pmod{2}.$$ 

{\it Case 2}: $2\leq i_0\leq q-2$ is even and $ 0\leq i_{1}\leq \frac{q^{(m-2)/2}-1}{q-1}-1$. By Equation (\ref{EQ::3}), we get that 
\begin{align*}
(b+ai) \bmod{n}&= q^{m-1}+(q^{(m-2)/2}-1)i_1+(q-1-\tfrac{i_0}{2}) (\tfrac{q^m-1}{q-1})+(\tfrac{q^{(m-2)/2}-1}{q-1})i_0\\
	&=[q^{m/2}+(q-1-\tfrac{i_0}{2})(\tfrac{q^{(m+2)/2}-1}{q-1})+i_1+1]q^{(m-2)/2}+\tfrac{i_0}{2}(\tfrac{q^{(m-2)/2}-1}{q-1})-i_1-1\\
	&=q^{m-1}+(q-1-\tfrac{i_0}{2})q^{m-2}+(q-1-\tfrac{i_0}{2})q^{m-3}+(q^{(m-2)/2}-1-j)q^{(m-2)/2}+j,
\end{align*}
where $j=\tfrac{i_0}{2}(\tfrac{q^{(m-2)/2}-1}{q-1})-i_1-1$. It then follows that 
\begin{align*}
\wt_q( (b+ai) \bmod{n})&=1+2(q-1-\tfrac{i_0}2)+\wt_q(q^{(m-2)/2}-1-j)+\wt_q(j)\\
&=1+(q-1)(\tfrac{m+2}{2})-i_0\\
&\equiv 0\pmod{2}.	
\end{align*}

{\it Case 3}: $1\leq i_0\leq q-2$ is odd and $ 0\leq i_{1}\leq \frac{q^{(m-2)/2}-1}{q-1}$. By Equation (\ref{EQ::3}), we get that 
\begin{align*}
(b+ai) \bmod{n}&= q^{m-1}+(q^{(m-2)/2}-1)i_1+(\tfrac{q-1-i_0}{2}) (\tfrac{q^m-1}{q-1})+(\tfrac{q^{(m-2)/2}-1}{q-1})i_0\\
	&=[q^{m/2}+(\tfrac{q-1-i_0}{2})(\tfrac{q^{(m+2)/2}-1}{q-1})+i_1]q^{(m-2)/2}+(\tfrac{q-1+i_0}{2})(\tfrac{q^{(m-2)/2}-1}{q-1})-i_1\\
	&=q^{m-1}+(\tfrac{q-1-i_0}{2})q^{m-2}+(\tfrac{q-1-i_0}{2})q^{m-3}+(q^{(m-2)/2}-1-j)q^{(m-2)/2}+j,
\end{align*}
where $j=(\tfrac{q-1+i_0}{2})(\tfrac{q^{(m-2)/2}-1}{q-1})-i_1$. It then follows that 
$$\wt_q((b+ai) \bmod{n} )=(q-1)(\tfrac{m}2)+1-i_0\equiv 0\pmod{2}.$$

In summary, $ \wt_{q}((b+ai) \bmod{n})\equiv0\pmod{2} $ for each $ 1\leq i\leq q^{(m-2)/2} $. This completes the proof.
\end{proof}

\begin{lemma}\label{lem:14}
	Let $m \equiv 0\pmod{4}$ and $m\geq 4$. Let $q=2^s\geq 4$, $n=q^m-1$ and $a=\frac{q^{m}-1}{q-1}-2(\frac{q^{(m+2)/2}-1}{q-1})$. Then $\gcd(a,n)=1 $  and 
	$$
\{ai \bmod n:~1\leq i\leq q^{(m-2)/2} \} \subseteq T_{(q,m; 1)}. 
	$$
\end{lemma}
\begin{proof}
It is clear that $ a  \equiv m-2(\tfrac{m+2}{2})\equiv -2\pmod{q-1}
$. Hence, $ \gcd(a, q-1)=1 $. Notice that $q=2^s$, $(q^m-1)/(q-1)\equiv 1\pmod{2}$. Consequently, 
\begin{align*}
\gcd(a, q^{m}-1) &= \gcd(a,\tfrac{q^{m}-1}{q-1})\\
&= \gcd(\tfrac{q^{{(m+2)/2}}-1}{q-1},\tfrac{q^{m}-1}{q-1})\\
&= \frac{\gcd(q^{{(m+2)/2}}-1,q^{m}-1)}{q-1}\\
&= \frac{q^{\gcd({(m+2)/2},m)}-1}{q-1}\\
&=1.
\end{align*}

Let $ i=(q-1)i_{1}+i_{0} $, where $ 0\leq i_{0}\leq q-2 $, $ 0\leq i_{1}\leq \frac{q^{(m-2)/2}-1}{q-1}$ and $i_0\in \{0,1\}$ if $i_1=\frac{q^{(m-2)/2}-1}{q-1}$. It is easily verified that
\begin{equation}\label{EQ::4}
ai\equiv [i_{0}(\tfrac{q^{(m-2)/2}-1}{q-1})-2i_{1}]q^{(m+2)/2}+2i_{1}-i_{0}(\tfrac{q^{(m+2)/2}-1}{q-1})	 \pmod{n}.
\end{equation}
Notice that $\tfrac{q^{(m-2)/2}-1}{q-1}\equiv 1\pmod{2}$, $i_{0}(\tfrac{q^{(m-2)/2}-1}{q-1})-2i_{1}=0$ if and only if $i_0$ is even and $i_1=\tfrac{i_0}2( \tfrac{q^{(m-2)/2}-1}{q-1})$. Therefore, if $1\leq i\leq q^{(m-2)/2}$, $i_{0}(\tfrac{q^{(m-2)/2}-1}{q-1})-2i_{1}\neq 0$. The rest of the proof is divided into the following cases.
 
{\it Case 1}: $i_{0}(\tfrac{q^{(m-2)/2}-1}{q-1})-2i_{1}>0$. In this case, $i_0\geq 1$ and $2i_{1}-i_{0}(\tfrac{q^{(m+2)/2}-1}{q-1})<0 $. Therefore,
\begin{align*}
ai \bmod{n}&=[i_{0}(\tfrac{q^{(m-2)/2}-1}{q-1})-2i_{1}-1]q^{(m+2)/2}+q^{(m+2)/2}+2i_{1}-i_{0}(\tfrac{q^{(m+2)/2}-1}{q-1})	\\
&= j q^{(m+2)/2}+(q-1-i_0)q^{m/2}+(q-1-i_0)q^{(m-2)/2}+q^{(m-2)/2}-1-j,
\end{align*}
where $j=i_{0}(\tfrac{q^{(m-2)/2}-1}{q-1})-2i_{1}-1$. It follows that 
\begin{align*}
\wt_q( ai \bmod{n} )&=\wt_q(j)+2(q-1-i_0)+\wt_q( q^{(m-2)/2}-1-j )\\
&=2(q-1-i_0)+(q-1)(\tfrac{m-2}{2})\\
&\equiv 1\pmod{2}.	
\end{align*}

{\it Case 2}: $i_{0}(\tfrac{q^{(m-2)/2}-1}{q-1})-2i_{1}<0$. By $ 0\leq i_{0}\leq q-2 $ and $ 0\leq i_{1}\leq \frac{q^{(m-2)/2}-1}{q-1}$, we deduce that $i_0\in \{0,1\}$. The proof of this case is divided into the following two subcases
 \begin{itemize}
 \item $i_0=0$ and $1\leq i_1\leq \tfrac{q^{(m-2)/2}-1}{q-1}$. By Equation (\ref{EQ::4}),
 $$(-ai) \bmod{n}=(2i_1-1)q^{(m+2)/2}+q^{(m+2)/2}-1-(2i_1-1).$$
 It follows that 
 \begin{align*}
 \wt_q( ai \bmod{n} )&=(q-1)m -\wt_q((-ai)\bmod{n})\\
 &=(q-1)m-[\wt_q(2i_1-1)+(q-1)(\tfrac{m+2}2)-\wt_q(2i_1-1)]\\
 &=(q-1)(\tfrac{m-2}2)\\
 &\equiv 1\pmod{2}.	
 \end{align*}
\item $i_0=1$ and $\frac{q^{(m-2)/2}-1}{2(q-1)}\leq i_{1}\leq \frac{q^{(m-2)/2}-1}{q-1}$. Then 
$$ (-ai)\bmod{n}=j q^{(m+2)/2}+q^{m/2}+q^{(m-2)/2}-1-(j-1), $$
where $j=2i_1-\tfrac{q^{(m-2)/2}-1}{q-1}$. It follows that 
\begin{align*}
\wt_q( ai \bmod{n} )&=(q-1)m-\wt_q((-ai)\bmod{n} )\\
&=(q-1)m-[\wt_q(j)+1+(q-1)(\tfrac{m-2}2)-\wt_q(j-1)]\\
&=(q-1)(\tfrac{m+2}2)-1+\wt_q(j-1)-\wt_q(j).
\end{align*}
Notice that $1\leq j\leq \tfrac{q^{(m-2)/2}-1}{q-1}$ is odd, then $j=j_0+j_1q+\cdots+j_{(m-4)/2}q^{(m-4)/2}$, where $0\leq j_h\leq q-1$ and $j_0$ is odd. Consequently, $j-1=j_0-1+j_1q+\cdots+j_{(m-4)/2}q^{(m-4)/2}$. Therefore, $\wt_q(j-1)=\wt_q(j)-1$. It then follows that 
$$\wt_q(a i \bmod{n})=(q-1)(\tfrac{m+2}2)-2\equiv 1\pmod{2}.$$ 
 \end{itemize}

In summary, $ \wt_{q}(ai \bmod{n}) \equiv 1\pmod{2}$ for each $ 1\leq i\leq q^{(m-2)/2} $. This completes the proof.
\end{proof}	

By the BCH bound and Lemmas \ref{lem:13} and \ref{lem:14}, we obtain the following results.

\begin{theorem}\label{thm:15}
Let $m\equiv 0\pmod{4}$ and $m\geq 4$. Let $q=2^s\geq 4$ and $n=q^m-1$. Then $d(\C_{(q, m; 0)})\geq  q^{(m-2)/2}+1$ and  $d(\C_{(q, m; 1)})\geq q^{(m-2)/2}+1$	.
\end{theorem}
 
 \begin{proof}
 The desired lower bound on the minimum distance of 	$\C_{(q, m; 0)}$ (resp.	$\C_{(q, m; 1)}$ ) follows from the BCH bound and Lemma \ref{lem:13} (resp. Lemma \ref{lem:14}). This completes the proof.  
 \end{proof}
 
 Summarizing the results in Theorems \ref{thm:2} and \ref{thm:8}, we arrive at the following conclusion.

\begin{theorem}\label{thm:16}
	Let $m\geq 3$ be odd, $q=2^s$ and $n=q^m-1$. The following hold.
	\begin{enumerate}
\item The pair of $q$-ary odd-like duadic codes $\C_{(q, m; 0)}$ and $\C_{(q, m; 1)}$ have parameters 
$$[n, (n+1)/2, d\geq q^{(m-1)/2}+2q-1].$$ 	
\item The $q$-ary self-dual code $\overline{\C_{(q, m; i)}}$ has parameters $[n+1, (n+1)/2, d\geq q^{(m-1)/2}+2q-1]$.
\item The $q$-ary self-orthogonal code $\widetilde{\C_{(q, m; i)}}$ has parameters $[n, (n-1)/2, d\geq q^{(m-1)/2}+2q-1]$.
\end{enumerate}
\end{theorem}

\begin{proof}
The desired results on dimensions of three codes directly follow from Theorem \ref{thm:2}. Clearly, $d(\overline{\C_{(q, m; i)}} )\geq d(\C_{(q, m; i)})$ and  $d(\widetilde{\C_{(q, m; i)}} )\geq d(\C_{(q, m; i)})$. The desired results on the minimum distances of three codes directly follow from Theorem \ref{thm:8}. This completes the proof. 
\end{proof}

When $q=2^s$, $m$ is odd and $n=q^m-1$, the $q$-ary punctured Reed-Muller codes $\GRM(m,[m(q-1)-1]/2)^*$ are duadic codes with parameters $[n, (n+1)/2, (1/2)(q+2)q^{(m-1)/2}-1 ]$ (see \cite{smid}). The following example illustrates that $\C_{(q, m; i)}$ and $\GRM(m,[m(q-1)-1]/2)^*$ are different.

\begin{example}
Let $q=4$, $m=3$ and $n=q^m-1=63$. Let $\beta$ be a primitive element of $\gf(q^3)$ with $\beta^3+\beta^2+\beta+\omega$, where $\omega$ is a primitive element of $\gf(4)$. Then $\C_{(q, m; 0)}$ has generator polynomial $x^{31}+x^{30}+\omega^2*x^{29} + x^{27} + \omega^2 x^{26} + \omega^2 x^{25} + \omega^2 x^{24} + x^{23} +x^{21} + x^{18} + \omega x^{17} + \omega x^{16} + x^{15} + x^{13} + x^{12} + x^{10} + x^9 + x^8 +\omega x^7 + \omega^2 x^6 + \omega x^5 + \omega x^4 + \omega^2 x^3 + \omega x^2 + 1$, and $\C_{(q, m; 1)}$ has generator polynomial $x^31 + \omega x^{29} + \omega^2 x^{28} + \omega x^{27} + \omega x^{26} + \omega^2 x^{25} + \omega x^{24} + x^{23} + x^{22} + x^{21} + x^{19} + x^{18} + x^{16} + \omega x^{15} + \omega x^{14} + x^{13} + x^{10}+ x^8 + \omega^2 x^7 + \omega^2 x^6 + \omega^2 x^5 + x^4 + \omega^2 x^2 + x + 1$, respectively. These two codes have parameters $[63,32,15]$. The best quaternary linear code of length $63$ and dimension $32$ has minimum distance $16$  \cite{Grassl}. The $q$-ary punctured Reed-Muller code $\GRM(3,4)^*$ has parameters $[63,32,11]$. Moreover, the following hold.
\begin{itemize}
\item The extended code $\overline{\C_{(q, m; i)}}$ is a quaternary self-dual code with parameters $[64,32,16]$. The best quaternary linear code of length $64$ and dimension $32$ has minimum distance $17$  \cite{Grassl}. 
\item The code $\widetilde{\C_{(q, m; i)}}$	is a quaternary self-orthogonal code with parameters $[63, 31,16]$. The best quaternary linear code of length $63$ and dimension $31$ has minimum distance $18$ \cite{Grassl}.
\end{itemize}
\end{example}

 Summarizing the results in Theorems \ref{thm:3}, \ref{thm:12} and \ref{thm:15}, we arrive at the following conclusion.

\begin{theorem}\label{thm:18}
	Let $m\geq 2$ be even, $q=2^s$ and $n=q^m-1$. Then the following hold.
	\begin{enumerate}
	\item The $q$-ary LCD code $\C_{(q, m; 0)}$ has parameters $[n, (n+3)/2, d\geq d_{(q, m;0)}]$, where	
	$$d_{(q, m; 0)}=\begin{cases}
 (q+2)/2~&{\rm if}~m=2,\\
 2q^2-2q+2~&{\rm if}~m=6,\\
 q^{(m-2)/2}+2q-1~&{\rm if}~m\geq 10~{\rm and}~m\equiv 2~(~{\rm mod}~4),\\
 q^{(m-2)/2}+1~&{\rm if}~m\equiv 0~(~{\rm mod}~4).
 \end{cases}	
 $$
	\item The $q$-ary LCD code $\C_{(q, m; 1)}$ has parameters $[n, (n-1)/2, d\geq d_{(q, m;1)}]$, where	
	$$d_{(q, m; 1)}=\begin{cases}
 (q+2)/2~&{\rm if}~m=2,\\
 q^{(m-2)/2}+2q-1~&{\rm if}~m\geq 6~{\rm and}~m\equiv 2~(~{\rm mod}~4),\\
 q^{(m-2)/2}+1~&{\rm if}~m\equiv 0 ~(~{\rm mod}~4).
 \end{cases}	
 $$
	\end{enumerate}
\end{theorem}

When $m$ is even and $q=4$, Shi et al. \cite{STKS2023} raised the following open problem: How to estimate the parameters of $\C_{(q, m; 1)}$? Theorem \ref{thm:18} answers this open question.

\section{Concluding remarks}\label{sec:4} 

The main contribution of this paper is the analysis of the $2^s$-ary Tang-Ding codes. A good lower bound on the minimum distance of the $2^s$-ary Tang-Ding codes was presented. Several infinite families of $2^s$-ary cyclic codes with a square-root-like lower bound were constructed. A summary of the main specific contributions of this paper goes as follows. 
\begin{enumerate}
\item A pair of $2^s$-ary odd-like duadic codes with parameters $$[2^{sm}-1, 2^{sm-1}, d\geq 2^{s(m-1)/2}+2^{s+1}-1]$$ was constructed, where $s\geq 2$ and $m\geq 3$ is odd. By expanding them, an infinite family of $2^s$-ary self-dual codes with parameters $[2^{sm}, 2^{sm-1}, d\geq 2^{s(m-1)/2}+2^{s+1}-1]$ was constructed (see Theorem \ref{thm:16}).
\item An infinite family of $2^s$-ary self-orthogonal codes with parameters 
$$[2^{sm}-1, 2^{sm-1}-1, d\geq 2^{s(m-1)/2}+2^{s+1}-1]$$ was constructed, where $s\geq 2$ and $m\geq 3$ is odd (see Theorem \ref{thm:16}).
\item An infinite family of $2^s$-ary LCD codes with parameters 
$$[2^{sm}-1, 2^{sm-1}+1, d\geq d_{(2^s,m;0)}]$$ was constructed, where $s\geq 2$, $m\geq 2$ is even, and 
$$d_{(2^s, m; 0)}=\begin{cases}
 2^{s-1}+1~&{\rm if}~m=2,\\
 2^{2s+1}-2^{s+1}+2~&{\rm if}~m=6,\\
 2^{s(m-2)/2}+2^{s+1}-1~&{\rm if}~m\geq 10~{\rm and}~m\equiv 2~(~{\rm mod}~4),\\
 2^{s(m-2)/2}+1~&{\rm if}~m\equiv 0~(~{\rm mod}~4).
 \end{cases}	
 $$
 (see Theorem \ref{thm:18}).
 
 \item An infinite family of $2^s$-ary LCD codes with parameters 
$$[2^{sm}-1, 2^{sm-1}-1, d\geq d_{(2^s,m;1)}]$$ was constructed, where $s\geq 2$, $m\geq 2$ is even, and 
$$d_{(2^s, m; 1)}=\begin{cases}
 2^{s-1}+1~&{\rm if}~m=2,\\
  2^{s(m-2)/2}+2^{s+1}-1~&{\rm if}~m\geq 6~{\rm and}~m\equiv 2~(~{\rm mod}~4),\\
 2^{s(m-2)/2}+1~&{\rm if}~m\equiv 0~(~{\rm mod}~4).
 \end{cases}	
 $$
 (see Theorem \ref{thm:18}).
\end{enumerate}

\end{document}